\documentclass[%
reprint,
nofootinbib,
 amsmath,amssymb,
 aps,
pra,
floatfix,
]{revtex4-2}

\usepackage{graphicx}% Include figure files
\usepackage{dcolumn}% Align table columns on decimal point
\usepackage{bm}% bold math
\usepackage{graphicx}
\usepackage{amsmath}
\usepackage{amsthm}
\usepackage[plain, noend]{algorithm}
\usepackage{algpseudocode}
\usepackage{stmaryrd}
\usepackage{enumitem}
\usepackage{float}
\usepackage{thmtools}
\usepackage{cleveref}
\usepackage{thm-restate}
\usepackage[all]{xy}
\usepackage{subcaption}
\usepackage{qcircuit}
\usepackage{tikz}
\usepackage{tabularx}
\usetikzlibrary{shapes.geometric} % Load shapes.geometric library
\usetikzlibrary{arrows.meta} % Load arrows.meta library
\usetikzlibrary{arrows.meta, positioning}

\newcommand{\Emp}{\textrm{Emp}\,}
\newcommand{\Dup}{\textrm{Dup}\,}

\newenvironment{steps}
{
    \begin{enumerate}[label=\textbf{Step \arabic*:}, leftmargin=*]
}
{
    \end{enumerate}
}

\newtheorem{theorem}{Theorem}[section]
\newtheorem{lemma}[theorem]{Lemma}
\newtheorem{corollary}[theorem]{Corollary}
\newtheorem{proposition}[theorem]{Proposition}
\newtheorem{definition}[theorem]{Definition}

\newtheorem{problem}[theorem]{Problem}

\begin{document}

\preprint{APS/123-QED}

\title{Minimum Synthesis Cost of CNOT Circuits}
% Force line breaks with \\Manuscript Title:\\with Forced Linebreak
\thanks{This document is a preprint uploaded to arXiv on August 15th, 2024.}%

\author{Alan Bu}
 \email{abu@college.harvard.edu}
 \affiliation{Harvard University, Cambridge, MA 02138} %\\
\author{Evan Fan}%
 \email{etfan@exeter.edu}
\affiliation{Phillips Exeter Academy, Exeter, NH 03833}%
\author{Robert Joo}
 \email{sjoo@exeter.edu}
 \affiliation{Phillips Exeter Academy, Exeter, NH 03833}

%\affiliation{
%\affiliation{
%\affiliation{%


\date{\today}% It is always \today, today,

\begin{abstract}
Optimizing the size and depth of CNOT circuits is an active area of research in quantum computing and is particularly relevant for circuits synthesized from the Clifford + T universal gate set. Although many techniques exist for finding short syntheses, it is difficult to assess how close to optimal these syntheses are without an exponential brute-force search. We use a novel method of categorizing CNOT gates in a synthesis to obtain a strict lower bound computable in $O(n^{\omega})$ time on the minimum number of gates needed to synthesize a given CNOT circuit, where $\omega$ denotes the matrix multiplication constant and $n$ is the number of qubits involved. Applying our framework, we prove that $3(n-1)$ gate syntheses of the $n$-cycle circuit are optimal and provide insight into their structure. We also generalize this result to permutation circuits.

For linear reversible circuits with $ n = 3, 4, 5$ qubits, our lower bound is optimal for 100\%, 67.7\%, and 23.1\% of circuits and is accurate to within one CNOT gate in 100\%, 99.5\%, and 83.0\% of circuits respectively. We also introduce an algorithm that efficiently determines whether certain circuits can be synthesized with fewer than $n$ CNOT gates.
\end{abstract}

\maketitle

\section{Introduction}

Minimizing the number of CNOT gates necessary to synthesize a given quantum circuit is a well-studied task in quantum computing. The authors of \cite{barenco1995elementary, divincenzo1995two} have developed algorithms to optimize the number of CNOT gates in quantum circuits, and \cite{patel2003efficient, vatan2004optimal, shende2004smaller, shende2005synthesis, vartiainen2004efficient, knill1995approximation} give asymptotics on the minimum number of CNOT gates required in circuit synthesis. The approaches presented in these papers have been used to minimize both the \emph{size} and \emph{depth} of a circuit, in some cases accounting for the qubit connectivity of the quantum computer and the ancilla count. In general, these methods have approached asymptotic optimality \cite{nash2020quantum, jiang2020optimal, vatan2004optimal,  kang2023cnot}. 

Minimizing the number of CNOT operations is important because CNOT operations and single qubit gates form a universal set for quantum computing. Both the speed and error rate of quantum algorithms such as Shor's algorithm and Grover's search improve with decreases in CNOT gate count \cite{liu2023minimizing, brickman2005implementation}. 

\subsection{Main Results}

Although a large volume of active research has been dedicated to optimizing CNOT circuits, we know little about how to obtain firm lower bounds on the minimum number of CNOT gates required to synthesize a given circuit besides resorting to brute-force search. Consequently, an exponential search time is required to verify the optimality of a CNOT circuit decomposition. In this paper, we focus on quantum computers without topological constraints or quantum reconfigurability. Our main result utilizes a novel categorization of the CNOT gates in a CNOT decomposition into link, middle, and cut gates to find strict lower bounds in $O(n^{\omega})$ time on the minimum number of CNOT gates necessary to synthesize a given circuit, where $\omega$ denotes the matrix multiplication constant \cite{duan2023faster} and $n$ is the number of qubits involved. We apply this method to all $3, 4, 5$-qubit circuits. In these cases, our lower bound is exact in $100\%, 67.7\%,$ and $23.1\%$ of circuits, and off from the minimum by at most one in $100\%, 99.5\%$, and $83.0\%$ of circuits respectively. 

Notably, \cite{liu2024realization, planat2017magic, bataille2022quantum, vatan2004optimal} provided constructions which use $3(n-1)$ CNOT gates to synthesize an $n$-cycle circuit, but did not disprove the existence of syntheses with fewer CNOT gates. By applying our lower bound, we prove that $3(n-1)$ gates is optimal, eliminating the need for further CNOT count optimization of $n$-cycle circuits in future research. More generally, we show the minimum number of CNOT gates needed to synthesize permutation circuits is $3(n-k)$, where $k$ is the number of disjoint cycles in the permutation. We remark that our framework not only yields a certificate of optimality, but also shows that the gates in any $3(n-1)$ gate synthesis of an $n$-cycle circuit can have its CNOT gates partitioned into three sets of size $n - 1$, each of which forms a spanning tree over the $n$ qubits, all without needing to resort to brute-force search.

Finally, Jiang et al. \cite{jiang2020optimal} proved that determining the minimum CNOT count of an arbitrary CNOT circuit with topological constraints is $\textrm{NP}$-hard. Finding the computational complexity of the same problem without topological constraints is an ongoing area of research \cite{kang2023cnot, murray2017non, amy2018controlled}. We introduce the linkability problem, a subclass of this problem, with the goal of efficiently determining whether certain circuits can be synthesized with fewer than $n$ CNOT gates. Combining our framework with a graph-theoretic approach, we give an $O(n^\omega)$ algorithm for this problem.

\subsection{Organization of the Paper}

In Section \ref{preliminaries}, we introduce conventions and definitions used throughout the paper. In Section \ref{allgates}, we introduce the categorization of CNOT gates into link, middle, and cut gates, as well as derive a lower bound on the number of link and cut gates necessary to synthesize a circuit. In Section \ref{middlegates}, we prove results regarding the effects of middle gates on a circuit and derive a lower bound on the number of middle gates necessary to synthesize a circuit. In Section \ref{mainalgo}, we combine all of these results alongside a few additional observations to yield our lower bound and apply the bound to $n$-cycle and permutation circuits to determine their minimum CNOT gate count. In Section \ref{linkability}, we introduce and prove our polynomial-time algorithm for the linkability problem. In Section \ref{results}, we assess the performance of our lower bound in the $n = 3, 4, 5$ cases. In Section \ref{conclusion}, we summarize our main results and raise future directions.

\section{Preliminaries and Notations}
\label{preliminaries}
We use the notation $[n] = \{1,2,\ldots,n\}$. We use $M_{i,j}$ to denote the value in the $i$th row and $j$th column of a matrix $M$. We use $e$ to denote the identity permutation and $T_{ij}$ to denote the $(i, j)$-transposition. We denote the identity matrix by $\mathbf{I}$. For a given set $A$ and permutation $\sigma$, we let $\sigma(A)$ denote the image of $A$ under $\sigma$. We let $\textbf{1}_{i \in A}$ denote the indicator function, where $\textbf{1}_{i \in A} = 1$ if $i \in A$ and $\textbf{1}_{i \in A}=0$ otherwise. We say that an $n \times n$ matrix $M$ is a \emph{reversible binary matrix} if $M \in GL_n(\mathbb{F}_2)$. We define $\Emp M$ to be the number of all-zero rows in $M$ and $\Dup M$ to be the number of disjoint pairs of duplicate, not all-zero, rows in $M$.

\subsection{CNOT Synthesis of Circuits}

In this paper, we lower bound the minimum number of CNOT gates required to synthesize a given linear reversible circuit, which can be represented as an $n \times n$ reversible binary matrix. Rigorously, we define the problem as follows, borrowing notation from \cite{jiang2020optimal}.

\begin{definition}
    A \emph{CNOT gate} or \emph{CNOT operation} takes as input two binary values $x_i$ and $x_j$ representing the \emph{control} and \emph{target} of the CNOT gate, respectively, and replaces the pair $(x_i,x_j)$ with $(x_i, x_i \oplus x_j)$. The effect of a CNOT gate with control $i$ and target $j$ on a set of $n$ qubits can be represented as multiplying by matrix $M$ with $M_{i,j} = 1$, all diagonal entries equal to $1$, and all other entries equal to $0$. We denote this CNOT gate, which adds the $i$th row of a matrix to the $j$th row, as $\textrm{CNOT}(i, j)$.
\end{definition}

\begin{definition}
    For any $n \times n$ reversible binary matrix $M$, we define its \emph{size} $s(M)$ to be one less than the length of the shortest sequence $M_1, M_2, \ldots, M_{k+1}$ of $n \times n$ reversible binary matrices with $M_1 = \textbf{I}$ and $M_{k+1} = M$ such that $M_{i+1}$ can be obtained from $M_i$ through a single CNOT operation for all $i \in [k]$.\footnote{We note that our usage of the term ``size'' differs from its conventional definition, where it typically refers to the number of CNOT gates in a given CNOT synthesis. In this paper, when we refer to the ``size'' of a circuit, we mean the minimum number of CNOT gates required in any synthesis of the circuit.} We call any such minimal sequence a \emph{CNOT synthesis} of $M$. Furthermore, this sequence gives a decomposition of $M$ as a product of CNOT gate matrices, which we call a \emph{CNOT decomposition} of $M$. An example is shown in \Cref{fig:decomposition}. 
\end{definition}

\begin{figure*}[!hbtp]
    \centering
    \begin{subfigure}{.13\textwidth} % Adjusted width for matrix subfigure
        \centering
        \renewcommand{\arraystretch}{1.5} % Adjust row height within matrix
        \scalebox{1.3}{
        $\begin{bmatrix}
            1 & 0 & 0 \\
            0 & 1 & 0 \\
            0 & 0 & 1
        \end{bmatrix}$
        }
         \vspace{0.18cm}
        \caption*{Initial matrix}
    \end{subfigure}
    \hspace{0.02\textwidth} % Adjust horizontal spacing for arrow
    \begin{subfigure}{0.02\textwidth}
        \hspace{-0.63cm}
        \scalebox{1.7}{
        $\to$
        }
        \vspace{1.32cm}
    \end{subfigure}
    \hspace{0.02\textwidth} % Adjust horizontal spacing before circuit
    \begin{subfigure}{.3\textwidth} % Adjusted width for circuit subfigure
        \centering
        \[
        \Qcircuit @C=0.8em @R=1.5em {
            & x_1 \quad \quad &\ctrl{1} & \ctrl{2} & \targ & \ctrl{1} & \qw & \qw & \qw & \; x_2\\
            & x_2 \quad \quad & \targ & \qw & \ctrl{-1} & \targ & \targ & \ctrl{1} & \qw & \; x_3\\
            & x_3 \quad \quad & \qw & \targ & \qw & \qw & \ctrl{-1} & \targ & \qw & \; x_1
        }
        \]
        \caption*{CNOT decomposition}
    \end{subfigure}
    \hspace{0.02\textwidth} % Adjust horizontal spacing for arrow
    \begin{subfigure}{0.02\textwidth}
        \hspace{-0.65cm}
        \scalebox{1.7}{
        $\to$
        }
        \vspace{1.32cm}
    \end{subfigure}
    \hspace{0.02\textwidth} % Adjust horizontal spacing before final matrix
    \begin{subfigure}{.13\textwidth} % Adjusted width for matrix subfigure
        \centering
        \renewcommand{\arraystretch}{1.5} % Adjust row height within matrix
        \scalebox{1.3}{
        $\begin{bmatrix}
            0 & 1 & 0 \\
            0 & 0 & 1 \\
            1 & 0 & 0
        \end{bmatrix}$
        }
        \vspace{0.18cm}
        \caption*{$3$-cycle matrix}
    \end{subfigure}
    \caption{CNOT decomposition of the $3$-cycle permutation operation.}
    \label{fig:decomposition}
\end{figure*}
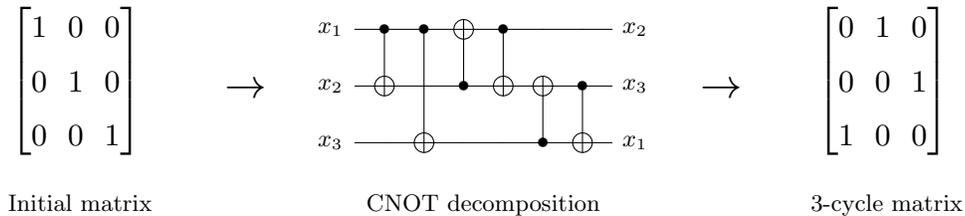

The authors of \cite{jiang2020optimal,maslov2022depth} also minimized the number of parallel CNOT operations required to synthesize a given circuit, as fully disjoint CNOT operations can be performed simultaneously on a quantum computer to save time.

\begin{definition}
    A \emph{parallel CNOT gate} or \emph{parallel CNOT operation} is a composition of arbitrarily many CNOT gates such that no qubit is operated on by more than one such gate. Rigorously, let $x_1, x_2, \ldots, x_k$ and $y_1, y_2, \ldots, y_k$ be $2k$ distinct numbers in $[n]$. A parallel CNOT gate can be represented by the matrix $M$ with $M_{x_i, y_i} = 1$ for $i \in [k]$, all diagonal entries equal to $1$, and all other entries equal to $0$.
\end{definition}

An example of a parallel CNOT decomposition is given in \Cref{fig:parallel}. 

\begin{figure*}[!hbtp]
    \centering 
    \begin{subfigure}{.30\textwidth}
        \centering
        \scalebox{1.3} {
        $\begin{bmatrix}
            0 & 1 & 0 & 0 & 0 & 0 \\ 
            1 & 0 & 1 & 0 & 0 & 1 \\ 
            0 & 0 & 1 & 0 & 0 & 0 \\ 
            0 & 0 & 1 & 1 & 0 & 0 \\ 
            0 & 0 & 0 & 0 & 1 & 1 \\ 
            0 & 1 & 1 & 0 & 0 & 1
        \end{bmatrix}$
        }
        \vspace{0.25cm}
    \end{subfigure} 
    \begin{subfigure}{.40\textwidth} 
        \centering 
        $
        \Qcircuit @C=1.0em @R=1.2em {
            & \lstick{x_1} & \qw & \ctrl{1} & \targ & \ctrl{5} & \qw & \qw & \qquad \qquad \; x_2 \\
            & \lstick{x_2} & \qw & \targ & \ctrl{-1} & \qw & \targ & \qw & \qquad \qquad \; x_1 \oplus x_3 \oplus x_6 \\
            & \lstick{x_3} & \qw & \ctrl{1} & \ctrl{3} & \qw & \qw & \qw & \qquad \qquad \; x_3 \\
            & \lstick{x_4} & \qw & \targ & \qw & \qw & \qw & \qw & \qquad \qquad \; x_3 \oplus x_4 \\
            & \lstick{x_5} & \qw & \targ & \qw & \qw & \qw & \qw & \rstick{x_5 \oplus x_6} \\
            & \lstick{x_6} & \qw & \ctrl{-1} & \targ & \targ & \ctrl{-4} & \qw & \qquad \qquad \; x_2 \oplus x_3 \oplus x_6
        }
        $
        \vspace{0.001cm}
    \end{subfigure}
    \caption{An illustration of a depth $4$ parallel CNOT decomposition.}
    \label{fig:parallel}
\end{figure*}
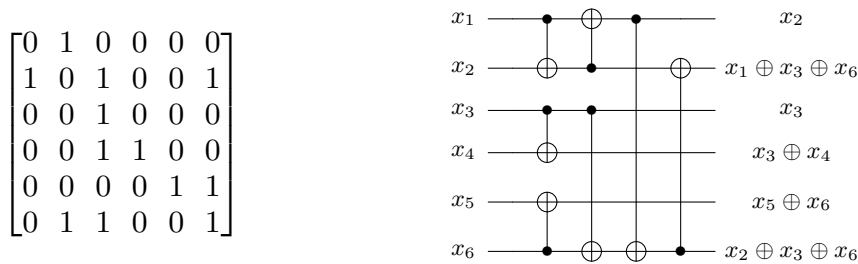

\begin{definition}
    For any $n \times n$ reversible binary matrix $M$, we define its \emph{depth} $d(M)$ to be one less than the length of the shortest sequence $M_1, M_2, \ldots, M_{k+1}$ of $n \times n$ reversible binary matrices with $M_1 = \textbf{I}$ and $M_{k+1} = M$ such that $M_{i+1}$ can be obtained from $M_i$ through a single parallel CNOT operation for all $i \in [k]$.\footnote{See above.}
\end{definition}

\subsection{Circuit-derived Graphs}

We further introduce two graphs fundamental to the lower bound introduced in this paper.

\begin{definition}
    Given an $n\times n$ reversible binary matrix $M$, we define its \emph{vertex-connectivity graph} to be $G_v(M) = (V,E)$, where $V = [n]$ and $(i,j)\in E$ if and only if $i \neq j$ and at least one of $M_{i,j}$ or $M_{j,i}$ is equal to one. We define $v(M)$ to be the number of connected components of this graph.
\end{definition}

\begin{definition}
    Given an $n\times n$ reversible binary matrix $M$, we define its \emph{edge-connectivity graph} to be $G_e(M) = (V,E)$, where $V = \{R_1,R_2,\ldots, R_n, C_1, \ldots, C_n\}$. For each row $R_i$ and column $C_j$ of $M$, we connect $R_i$ and $C_j$ with an edge in $G_e(M)$ if and only if $M_{i,j} = 1$. We define $e(M)$ to be the number of connected components of this graph. 
\end{definition}

Examples of the vertex and edge-connectivity graphs of a matrix are shown in \Cref{fig:susdef}.

\begin{figure*}[!hbtp]
    \centering 
    \hfill
    \begin{subfigure}{.25\textwidth}
        \centering
        \scalebox{1.3}{
        $\begin{bmatrix}
            0 & 0 & 1 & 0 & 0 \\ 
            0 & 1 & 0 & 1 & 0 \\ 
            1 & 0 & 0 & 0 & 0 \\
            0 & 0 & 0 & 0 & 1 \\
            0 & 1 & 0 & 0 & 0
        \end{bmatrix}$
        }
        \vspace{0.7cm}
        \caption{Matrix $M$}
    \end{subfigure}
    \hfill
    \begin{subfigure}{.35\textwidth}
        \centering 
        \begin{tikzpicture}[>={Stealth[width=2mm,length=2mm]}, node distance=1.5cm, every node/.style={circle, draw}]
            \node[regular polygon,draw=none, minimum size = 3.65 cm] (e) at (5,0) {};
             \foreach \x in {1,2,...,5}
            %\fill (e.corner \x) circle[radius=1.5pt];
            \node[] at (e.corner 1) (1) {1};
            \node[] at (e.corner 2) (2) {2};
            \node[] at (e.corner 3) (3) {3};
            \node[] at (e.corner 4) (4) {4};
            \node[] at (e.corner 5) (5) {5};
            \draw (1) -- (3);
            \draw (2) -- (4);
            \draw (2) -- (5);
            \draw (4) -- (5);
        \end{tikzpicture}
        \caption{$G_v(M)$}
    \end{subfigure}
    \begin{subfigure}{.35\textwidth}
        \centering 
        \begin{tikzpicture}[x=0.8cm,y=0.8cm]
            \foreach \x in {1,2,3,4,5}
            {
                \filldraw [black] (9,3-\x) circle (1.5pt);
                \node[anchor=east] at (9,3-\x) {$R_{\x}$};
                \filldraw [black] (11,3-\x) circle (1.5pt);
                \node[anchor = west] at (11,3-\x) {$C_{\x}$};
            }
            \draw (9,2) -- (11,0);
            \draw (9,1) -- (11,1);
            \draw (9,1) -- (11,-1);
            \draw (9,0) -- (11,2);
            \draw (9,-1) -- (11,-2);
            \draw (9,-2) -- (11,1);
        \end{tikzpicture}
        \caption{$G_e(M)$}
    \end{subfigure}
    \caption{An example of $M$ with $v(M) = 2$ and $e(M) = 4$.}
    \label{fig:susdef}
\end{figure*}
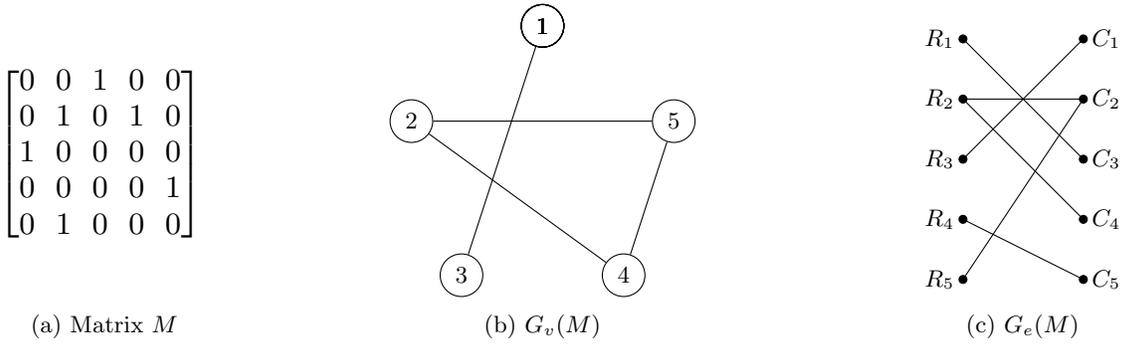

\begin{definition}
    We call a graph $G = (V,E)$ on $n$ vertices a \emph{star graph} if one vertex of $G$ has degree $n-1$ and all other vertices have degree $1$. We call $G$ a \emph{path graph} if we can label its vertices as $v_1, v_2, \ldots, v_n$ such that $(v_i, v_j) \in E$ if and only if $|i-j| = 1$. With slight abuse of notation, we say that a set of $n-1$ CNOT gates is a star or path if the graph $G$ formed by adding an edge between the control and target of each CNOT gate is a star or path graph respectively.
\end{definition}

\subsection{Influence Graphs}

Finally, we introduce two definitions used in our algorithm for the linkability problem.

\begin{definition}
    For a reversible binary $n \times n$ matrix $M$, we define its \emph{influence graph} $\mathcal{I}(M) = (V, A)$ to be the directed graph such that $V = [n]$ and $(i, j) \in A$ if and only if $M_{j,i} = 1$ and $i \neq j$. In each such case, we say that $i$ \emph{influences} $j$. An example is given in \Cref{fig:influences}.
\end{definition}

\begin{definition}
    For a directed graph $G = (V, A)$, we define a transitive reduction $G'$ of $G$ to be another directed graph with the same vertices and a minimal set of directed edges such there is a directed path from vertex $u$ to vertex $v$ in $G'$ if and only if there is a directed path from $u$ to $v$ in $G$.
\end{definition}

\begin{figure*}[!hbtp]
    \centering 
    % Adjust widths to approximately a third each to fit in one line
    \begin{subfigure}{.35\textwidth}
        \centering
        \scalebox{1.3}{
        $\begin{bmatrix} 
             1 & 0 & 1 & 1 \\ 
             1 & 1 & 1 & 0 \\ 
             0 & 0 & 1 & 0 \\ 
             0 & 0 & 1 & 1
        \end{bmatrix}$
        }
        \vspace{0.8cm}
        \caption{Matrix $M$}
    \end{subfigure}
    \hfill
   % Adds horizontal space between subfigures
    \begin{subfigure}{.28\textwidth}
        \centering
        \begin{tikzpicture}[>={Stealth[width=2mm,length=2mm]}, node distance=1.5cm, every node/.style={circle, draw}]
            % Define nodes
            \node (3) {3};
            \node[below left=of 3] (1) {1};
            \node[below right=of 3] (4) {4};
            \node[below right=of 1] (2) {2};

            % Draw edges
            \draw[->] (3) -- (1);
            \draw[->] (4) -- (1);
            \draw[->] (1) -- (2);
            \draw[->] (3) -- (2);
            \draw[->] (3) -- (4);
        \end{tikzpicture}
        \caption{Influence graph of $M$}
    \end{subfigure}
    \hfill
   % Adds horizontal space between subfigures
    \begin{subfigure}{.35\textwidth}
        \centering
        \begin{tikzpicture}[>={Stealth[width=2mm,length=2mm]}, node distance=1.5cm, every node/.style={circle, draw}]
            % Define nodes
            \node (3) {3};
            \node[below left=of 3] (1) {1};
            \node[below right=of 3] (4) {4};
            \node[below right=of 1] (2) {2};

            % Draw edges
            \draw[->] (4) -- (1);
            \draw[->] (1) -- (2);
            \draw[->] (3) -- (4);
        \end{tikzpicture}
        \caption{Transitive reduction of $\mathcal{I}(M)$}
    \end{subfigure}
    \caption{Example of an influence graph and its transitive reduction.}
    \label{fig:influences}
\end{figure*}
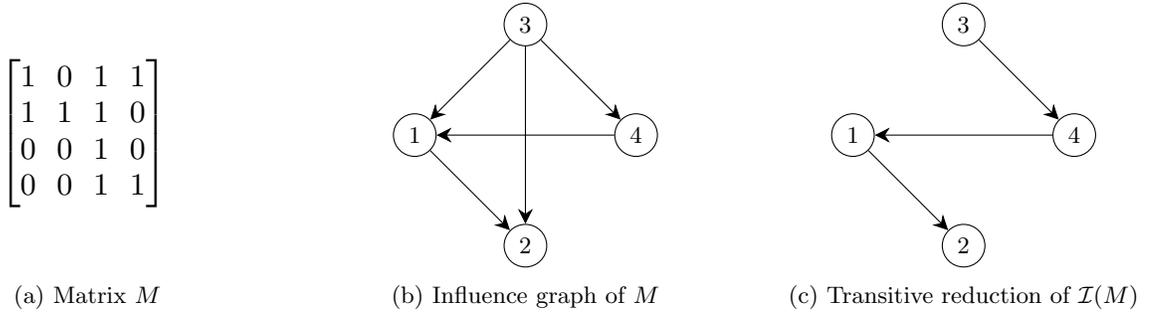

\section{Link, Middle, and Cut CNOT Gates}
\label{allgates}

The key insight of our algorithm is that each gate of a CNOT decomposition can be placed into one of three categories (or none), and finding lower bounds on the number of CNOT gates in each category allows us to obtain a lower bound on the total number of CNOT gates. In this section, we introduce these categories and prove lower bounds on the number of CNOT gates in two of the three categories.

\begin{definition}
    We define a \emph{river} of an $n \times n$ reversible binary matrix $M$ to be a permutation $\sigma \in S_n$ such that $M_{i, \sigma(i)} = 1$ for all $i$. We denote the set of all rivers of $M$ as $S(M)$. For any river $\sigma$, we denote the unique permutation matrix whose only nonzero entries lie on this river as $P_{\sigma}$.
\end{definition}

\begin{definition}
    Let $M, N$ be $n \times n$ reversible binary matrices such that $N$ can be obtained from $M$ through one CNOT operation. We call this CNOT operation a \begin{itemize}
        \item \emph{Link} gate if $v(M) > v(N)$,
        \item \emph{Middle} gate if $S(M) \neq S(N)$,
        \item \emph{Cut} gate if $e(M) < e(N)$.
    \end{itemize}
\end{definition}

We now prove an important result on the effects of link and cut gates on $G_v$ and $G_e$.

\begin{proposition}
    \label{etovclaim} Let $M$ be a reversible binary matrix. If $R_i, R_j$ are connected in $G_e(M)$, then $i, j$ are connected in $G_v(M)$.
\end{proposition}
    \begin{proof}
        As $R_i, R_j$ are connected in $G_e(M)$, there exists a path $R_{x_1} C_{y_1} R_{x_2} C_{y_2} \cdots C_{x_{m-1}} R_{x_m}$ of $G_e(M)$ such that $x_1 = i$ and $x_m = j$. Thus, $x_i, x_{i+1}$ are connected in $G_v(M)$ for all $i \in [m-1]$ since $M_{x_i, y_i} = M_{y_i, x_{i+1}} = 1$. Thus $i$ is connected to $j$ in $G_v(M)$ as desired.  
    \end{proof}

\begin{lemma}
    \label{linkdec}
    Let $M, N$ be $n \times n$ reversible binary matrices such that $N$ can be obtained from $M$ through one CNOT operation. Then, $|v(M) - v(N)| \leq 1$ and $|e(M) - e(N)| \leq 1$. Furthermore, if $v(N) = v(M) - 1$, then $e(N) = e(M) - 1$.  
\end{lemma}
\begin{proof} 
    Consider a CNOT operation on a matrix $M$ adding row $i$ to row $j$ for which $v(M)>v(N)$. If $i$ and $j$ were in the same connected component $A$ in $G_v(M)$, then any new edges in $G_v(N)$ will also be between vertices in $A$, contradicting $v(M)>v(N)$. If $i$ and $j$ were in distinct connected components $A$ and $B$ of $G_v(M)$, then the only edges that will be different in $G_v(N)$ will be those between vertices of $A$ and $B$. This will cause $A$ and $B$ to become a single connected component in $G_v(N)$, and leave all other connected components unchanged. Thus, if $v(M)>v(N)$ then $v(M)=v(N)+1$.

    Similarly, consider the case when $e(M)>e(N)$. If $R_i$, $R_j$ are in the same connected component $A$ of $G_e(M)$ then no edges between connected components are added, while if they are in separate connected components $A$ and $B$, then the only new edges are edges connecting a vertex in $A$ to a vertex in $B$, leaving all other connected components unchanged. Thus, if $e(M)>e(N)$ then $e(M)= e(N) +1$.

     Now, notice that applying $\textrm{CNOT}(i, j)$ on $N$ returns the original matrix $M$. In symmetric fashion, $v(N) \leq v(M) + 1$ and $e(N) \leq e(M) + 1$.
    
    Finally, if $v(N) = v(M) - 1$, then there must exist $i, j$ that lie in different connected components of $G_v(M)$ such that $N_{i, j} = 1$. The CNOT gate adds some row $k$ with $M_{k, j} = 1$ to row $i$. Note that $j, k$ lie in the same connected component of $G_v(M)$. After adding row $k$ to row $i$, we know $i, j, k$ lie in the same connected component of $G_v(N)$ and $R_i, C_j$ is connected in $G_e(N)$. Notably, $R_i$ and $C_j$ are in different connected components of $G_e(M)$ by \Cref{etovclaim}, and thus the number of edge-connected components in $N$ is one less than in $M$, as desired. 
\end{proof}

We now use the tools developed in the proof of \Cref{linkdec} to prove that these three categories of CNOT gates are mutually exclusive, allowing us to sum lower bounds on each category to bound the size of the entire circuit.

\begin{lemma}
    \label{lmc}    
    Let $M, N$ be $n \times n$ reversible binary matrices such that $N$ can be obtained from $M$ through one CNOT operation. Then the CNOT operation is either exclusively a link, middle, or cut gate, or none of these.
\end{lemma}

\begin{proof}
    By \Cref{linkdec}, if the CNOT operation is a link gate then $e(M) > e(N)$, so a link gate cannot be a cut gate.

    As discussed in the proof of \Cref{linkdec}, link gates operate on rows that are in separate connected components of $G_v(M)$. Let the control row of the CNOT gate lie in connected component $A \subseteq [n]$ of $G_v(M)$. Without loss of generality, let $A = [a]$ and define $B = [n] \setminus A$. 
    Then, $M = \begin{bmatrix}
    \textbf{A} & \textbf{0} \\
    \textbf{0} & \textbf{B}
    \end{bmatrix}$
    where \textbf{A} is an $a \times a$ matrix and \textbf{B} is an $(n-a) \times (n-a)$ matrix. Any $\sigma \in S(M)$ satisfies $\sigma(A) = A$. When we add a row from \textbf{A} to a row in \textbf{B}, the upper-right box of zeros does not change. Thus, for any $\sigma \in S(N)$, we have $\sigma(A) \subseteq A$, so $\sigma(A) = A$. Therefore the set of rivers is unchanged as no river passes through the newly created bottom left entries, so a link gate cannot be a middle gate.

    Finally, we show that cuts and middles are also disjoint. In this case, we can assume up to permutation of the rows and columns that $N = 
    \begin{bmatrix}
    \textbf{A} & \textbf{0} \\
    \textbf{0} & \textbf{B}
    \end{bmatrix}$
    where $\textbf{A}, \textbf{B}$ are not necessarily square. However, if $\textbf{A}, \textbf{B}$ are not square, then $N$ is not reversible, giving contradiction.
    Suppose the CNOT gate adds a row of \textbf{A} to a row of \textbf{B}. Let the set of columns of $M$ in \textbf{A} be denoted as $C_A$, and let the set of rows of $M$ in \textbf{A} be denoted as $R_A$. Thus if $\sigma \in S(N)$, then $\sigma(R_A) = C_A$. Furthermore, for any $\sigma \in S(M)$, we have $\sigma(R_A) \subseteq C_A$ so $\sigma(R_A) = C_A$. Thus the set of rivers in $M$ and $N$ are equal. This finishes the final case.
\end{proof}

Finally, we extend \Cref{linkdec} to yield a simple but powerful lower bound on the number of link and cut gates in a given circuit.

\begin{lemma}
    \label{linksandcuts}
    Any CNOT decomposition of $n \times n$ matrix $M$ must contain at least $n-v(M)$ link gates and at least $e(M)-v(M)$ cut gates.
\end{lemma}

\begin{proof}
    Consider any CNOT synthesis $M_1, M_2, \ldots, M_{k+1}$ of $M$. Note that $v(\textbf{I}) = e(\textbf{I}) = n$. The number of vertex-connected components must decrease over the sequence, so \Cref{linkdec} implies that at least $n-v(M)$ link gates are required. 

    From \Cref{linkdec}, link gates decrease the number of edge connected components by exactly one, while cut gates increase them by exactly one. Since there are $n$ initial components and $n-v(M)$ link gates, we need at least $e(M)-v(M)$ cut gates in order to end up at $e(M)$ edge-connected components in the final matrix.
\end{proof}

\section{Lower Bounds on Middle Gates}
\label{middlegates}
In this section, we develop an $O(n^{\omega})$ time algorithm that computes a lower bound on the number of middle gates necessary to synthesize a given circuit. To begin, we prove several lemmas that clarify the effects of middle gates on the rivers of a matrix. 

Given a set of middle CNOT gates, each river can only ``interact" with other rivers in the same equivalence class. 
%maybe can combine sentences?
These equivalence classes are determined by the set of middle CNOT gates. When it is clear that the rivers of the $M$ cannot be obtained from $\textbf{I}$ through these interactions, the given set of middle CNOT gates is insufficient to synthesize $M$. In the following subsections, we make this idea rigorous.

\subsection{Equivalence Classes of Rivers}
%\begin{definition}
    %Let $T$ denote the set of transpositions on $n$ elements. We define the Cayley graph generated by transpositions as $\Cay = (V, E)$ where $V = S_n$ and $(\sigma_1, \sigma_2) \in E$ if and only if $\sigma_1, \sigma_2 \in S_n$ differ by a transposition.
%\end{definition}

We now show that a singular middle CNOT gate only allows ``interactions" between certain pairs of rivers, each of which differs by a fixed transposition determined by the control and target of the CNOT gate.

\begin{lemma} 
    \label{creationannihilation}
    Let $M, N$ be $n \times n$ reversible binary matrices such that $N$ can be obtained from $M$ by adding row $i$ to row $j$. Consider pairing each permutation $\sigma$ to $T_{ij} \circ \sigma$. Then $|\{\sigma, T_{ij} \circ \sigma\} \cap S(M)| \equiv |\{\sigma, T_{ij} \circ \sigma\} \cap S(N)| \pmod 2$.
\end{lemma}

\begin{figure*}[!hbtp]
    \centering
    \scalebox{1.3}{
    \begin{subfigure}{.3\textwidth}
        \renewcommand{\arraystretch}{1.5}
        \centering
        \[
        \begin{matrix}
            1 & 0 \\
            0 & 1
        \end{matrix}
        \quad
        \xrightarrow{\hspace{2em}}
        \quad
        \begin{matrix}
            1 & 0 \\
            1 & 1
        \end{matrix}
        \]
    \end{subfigure}%
    }
    \scalebox{1.3}{
    \begin{subfigure}{.3\textwidth}
        \renewcommand{\arraystretch}{1.5}
        \centering
        \[
        \begin{matrix}
            1 & 1 \\
            0 & 1
        \end{matrix}
        \quad
        \xrightarrow{\hspace{2em}}
        \quad
        \begin{matrix}
            1 & 1 \\
            1 & 0
        \end{matrix}
        \]
    \end{subfigure}
    }
    \caption{Examples of the effects of $\textrm{CNOT}(i, j)$ for rows $i, j$ and columns $\sigma(i), \sigma(j)$.}
    \label{effectcnot}
\end{figure*}

\begin{proof}
    We begin with casework on the values of $M_{i, \sigma(i)}$ and $M_{i, \sigma(j)}$, with examples in \Cref{effectcnot}. If they are both $0$, then $\{\sigma, T_{ij} \circ \sigma\} \cap S(M) = \{\sigma, T_{ij} \circ \sigma\} \cap S(N) = \emptyset$ as desired. If exactly one of them is $1$, it can be checked that $\{\sigma, T_{ij} \circ \sigma\} \cap S(M) = \{\sigma, T_{ij} \circ \sigma\} \cap S(N)$ as desired. We now consider the case when $M_{i, \sigma(i)} = M_{i, \sigma(j)} = 1$.
    In this case, if it is not true that $M_{k, \sigma(k)} = 1$ for all $k \neq i, j$ then $\{\sigma, T_{ij} \circ \sigma\} \cap S(M) = \{\sigma, T_{ij} \circ \sigma\} \cap S(N) = \emptyset$ as desired. Otherwise, observe that $\sigma \in M$ if and only if $\sigma \notin N$ and $T_{ij} \circ \sigma \in M$ if and only if $T_{ij} \circ \sigma \notin N$. Thus, $(\{\sigma, T_{ij} \circ \sigma\} \cap S(M)) \sqcup (\{\sigma, T_{ij} \circ \sigma\} \cap S(N)) = \{\sigma, T_{ij} \circ \sigma\}$, completing the final case.
\end{proof}

\Cref{creationannihilation} demonstrates that when there is only one CNOT gate involved, the parity of the number of rivers in sets of the form $\{\sigma, T_{ij} \circ \sigma\}$ cannot change. We extend this idea to the case of multiple CNOT gates by allowing multiple transpositions.

\begin{definition}
    Let $S \subseteq T$ be a fixed set of transpositions. We say that $\sigma_1 \equiv_{S} \sigma_2$ if and only if $\sigma_2$ can be obtained from $\sigma_1$ through a sequence of transpositions in $S$. It is easy to see that $\equiv_S$ is an equivalence relation on $S_n$ as the reflexive, transitive, and symmetric properties hold. We let $S_n/S$ denote the set of equivalence classes under $\equiv_S$.
\end{definition}
 
\begin{theorem}
    \label{eqparity}
    Let $M, N$ be $n \times n$ reversible binary matrices such that $N$ can be obtained from $M$ by adding row $i$ to row $j$. Then, for any $S \subseteq T$ with $T_{ij} \in S$ and for any equivalence class $K \in S_n/S$, we have $|K \cap S(M)| \equiv |K \cap S(N)| \pmod 2$.
\end{theorem}
\begin{proof}
    Since $T_{ij} \in S$, we can partition all the rivers in $K$ into pairs $(\sigma, T_{ij} \circ \sigma)$ without including elements outside of $K$. By \Cref{creationannihilation},
\begin{align*}|K \cap S(M)| &= \sum |\{\sigma, T_{ij} \circ \sigma\} \cap S(M)| \\ &\equiv \sum |\{\sigma, T_{ij} \circ \sigma\} \cap S(N)| \\ &= |K \cap S(N)| \pmod 2.\end{align*}
\end{proof}

In preparation for a proof in the next subsection, we introduce the labeling of an equivalence class.

\begin{definition}
    Let $S \subseteq T$ and $K \in S_n/S$. Let $a_i = \min_{\sigma \in K} \sigma^{-1}(i)$. We define the \emph{labeling} $\mathcal{L}(K)$ of $K$ as  $(a_1, a_2, \ldots, a_n) \in [n]^n$. Additionally, when the context for $S$ is clear, we let $\mathcal{L}(\sigma)$ denote the labeling of the equivalence class of $\sigma$.
\end{definition}

In order to show that no two equivalence classes share a labeling, we prove the following lemma.

\begin{lemma}
    \label{treeconstruct}
    Suppose that there is a tree $T$ whose vertices are $v_1, v_2, \ldots, v_m$. We begin by choosing a set of initial values $p_0(v_i)$, and assigning $p(v_i) = p_0(v_i)$ for each vertex $v_i$. We are allowed to perform operations which exchange the values of $p(v_i)$ and $p(v_j)$ for adjacent vertices $v_i, v_j$. Then for any permutation $\sigma \in S_m$, there is a sequence of operations resulting in $p(v_i) = p_0(v_{\sigma(i)})$ for all $i \in [m]$.
\end{lemma}

\begin{figure*}[!hbtp]
\begin{tikzpicture}[x=0.6cm,y=0.6cm]
% First shape
 
\draw (-0.4, 0.8) -- (1,1) -- (2,0);
\filldraw [black] (-0.4, 0.8) circle (2pt);
\filldraw [black] (1,1) circle (2pt);
\filldraw [black] (1.2,2.4) circle (2pt) node[anchor = south]{$v_i$};
\filldraw [black] (-1.4,-0.2) circle (2pt);
\filldraw [black] (2,0) circle (2pt);
\filldraw [black] (-0.65, -1.4) circle (2pt);
\filldraw [black] (2.95,-2.6) circle (2pt) node[anchor=north]{$\sigma(i)$};
\filldraw [black] (2.75,-1.2) circle (2pt);
\filldraw [black] (-2.15, -1.4) circle (2pt);
\filldraw [black] (1.25,-1.2) circle (2pt);
\draw (1,1) -- (1.2,2.4);
\draw (-0.4, 0.8) -- (-1.4,-0.2);
\draw (2,0) -- (2.75,-1.2);
\draw (2,0) -- (1.25,-1.2);
\draw (-2.15, -1.4) -- (-1.4,-0.2);
\draw (-0.65, -1.4) -- (-1.4,-0.2);
\draw (2.95,-2.6) -- (2.75,-1.2);

\draw[->] (3.5,0) -- (4.5,0);
\draw[<->,thick] (3.2,-2.55) arc[start angle=-24, end angle=36,radius=1.4142];

% Second shape
\draw (6.6, 0.8) -- (8,1) -- (9,0);
\filldraw [black] (6.6, 0.8) circle (2pt);
\filldraw [black] (8,1) circle (2pt);
\filldraw [black] (8.2,2.4) circle (2pt) node[anchor = south]{$v_i$};
\filldraw [black] (5.6,-0.2) circle (2pt);
\filldraw [black] (9,0) circle (2pt);
\filldraw [black] (6.35, -1.4) circle (2pt);
\filldraw [black] (9.75,-1.2) circle (2pt) node[anchor=west]{$\sigma(i)$};
\filldraw [black] (9.95,-2.6) circle (2pt);
\filldraw [black] (4.85, -1.4) circle (2pt);
\filldraw [black] (8.25,-1.2) circle (2pt);
\draw (8,1) -- (8.2,2.4);
\draw (6.6, 0.8) -- (5.6,-0.2);
\draw (9,0) -- (9.75,-1.2);
\draw (9,0) -- (8.25,-1.2);
\draw (4.85, -1.4) -- (5.6,-0.2);
\draw (6.35, -1.4) -- (5.6,-0.2);
\draw (9.95,-2.6) -- (9.75,-1.2);

\draw[->] (10.5,0) -- (11.5,0);
\draw[->] (14.5,0) -- (15.5,0);
\draw[<->,thick] (9.95,-1.05) arc[start angle=0, end angle=60,radius=1.4142];

\filldraw [black] (12.6,0) circle (0.75pt);
\filldraw [black] (13,0) circle (0.75pt);
\filldraw [black] (13.4,0) circle (0.75pt);

% Third shape
\draw (17.6, 0.8) -- (19,1) -- (20,0);
\filldraw [black] (17.6, 0.8) circle (2pt);
\filldraw [black] (20.95,-2.6) circle (2pt);
\filldraw [black] (19.2,2.4) circle (2pt) node[anchor = south]{$v_i$};
\filldraw [black] (16.6,-0.2) circle (2pt);
\filldraw [black] (20,0) circle (2pt);
\filldraw [black] (17.35, -1.4) circle (2pt);
\filldraw [black] (19,1) circle (2pt);
\filldraw [black] (19.2,2.4) circle (2pt) node[anchor=west]{$\sigma(i)$};
\filldraw [black] (20.75,-1.2) circle (2pt);
\filldraw [black] (15.85, -1.4) circle (2pt);
\filldraw [black] (19.25,-1.2) circle (2pt);
\draw (19,1) -- (19.2,2.4);
\draw (17.6, 0.8) -- (16.6,-0.2);
\draw (20,0) -- (20.75,-1.2);
\draw (20,0) -- (19.25,-1.2);
\draw (15.85, -1.4) -- (16.6,-0.2);
\draw (17.35, -1.4) -- (16.6,-0.2);
\draw (20.95,-2.6) -- (20.75,-1.2);

\end{tikzpicture}
\caption{An illustration of \Cref{treeconstruct}.}
\label{treelem}
\end{figure*}

\begin{proof}
    Without loss of generality, we assume that $p_0(v_i) = i$ for $i \in [m]$. Consider the following algorithm: pick a leaf node $v_i$. We want to ensure that $p(v_i) = \sigma(i)$. We view $T$ as a rooted tree with $v_i$ as its root node and perform operations swapping $\sigma(i)$ up the tree until $p(v_i) = \sigma(i)$ (see \Cref{treelem}). Now, we can remove $v_i$ from $T$. We repeat this process until we have assigned $\sigma(i)$ to every $v_i$. 
\end{proof}

\begin{theorem}
\label{labeling}
    Let $S \subseteq T$. The function $\mathcal{L}: S_n/S \to [n]^n$ is injective.
\end{theorem}

\begin{proof}
    Let a labeling be $L = (a_1, a_2, \ldots, a_n)$, and suppose $\mathcal{L}(K) = L$ for some equivalence class $K \in S_n/S$. By definition, $a_i$ denotes the minimum $j$ such that $\sigma(j) = i$ over all $\sigma \in K$. Let $\sigma_0$ be a fixed element of $K$. If $T_{uv} \in S$, then $a_{\sigma_0(u)} = a_{\sigma_0(v)}$. Define $A_k = \{i | a_i = k \}$ for each $k \in [n]$. Let $G = (V, E)$ be a graph such that $V = [n]$ and $(i, j) \in E$ if and only if $T_{ij} \in S$.

    For each transposition $T_{uv} \in S$, both $u, v$ lie in the same set $\sigma_0^{-1}(A_k)$. Conversely, if $a_{\sigma_0(u)} = a_{\sigma_0(v)}$ then there exists a composition of transpositions in $S$ mapping $u$ to $a_{\sigma_0(u)}$ then to $v$. Thus each nonempty set $\sigma_0^{-1}(A_k)$ is a connected component of $G$.
    
    As a result, any $\sigma \in K$ must satisfy $\sigma(\sigma_0^{-1}(A_k)) = A_k$ for all $k$. In particular, the subgraph of $G$ induced by $\sigma_0^{-1}(A_k)$ is connected and has a spanning tree, so by \Cref{treeconstruct}, all permutations in $S_n$ satisfying this set of conditions must also lie in $K$. Hence $K$ is uniquely determined by $L$ and $S$, since the sets $\sigma_0^{-1}(A_k)$ are determined by $S$ as the connected components of $G$ and the sets $A_k$ are determined by $L$.
\end{proof}

Repeated applications of \Cref{eqparity} show that if $S$ contains the transpositions corresponding to each middle CNOT gate in a CNOT decomposition, then the parity of $|K \cap S(M)|$ is preserved throughout the entire CNOT decomposition. Although this condition is powerful, there are many possible equivalence classes $K$ based on the choice of CNOT gates, making it impractical to check the condition. In order to make the lower bound on middle gates computable in polynomial time, we should eliminate both the dependence on $K$ and $S(M)$.

\subsection{The $c_{\textrm{perfect}}$ Algorithm}

We introduce a consequence of \Cref{eqparity} which more directly allows us to find a lower bound on the number of middle gates without depending on the equivalence classes $K$ in \Cref{eqparity}.

\begin{theorem} 
    \label{abstractsum1}
    Let $M_1, M_2, \ldots, M_{k+1}$ be a CNOT synthesis of $n \times n$ matrix $M$. Let $M' = \textbf{I} + \sum_{\sigma \in S(M)} P_{\sigma}$. Let $G = (V, E)$ a graph such that $V = [n]$ and $(i, j) \in E$ if and only if there is some middle gate involved in the CNOT synthesis adding between rows $i, j$. Then for any connected component $A$ of $G$, let $\textbf{a} = (\textbf{1}_{1 \in A}, \textbf{1}_{2 \in A}, \ldots, \textbf{1}_{n \in A})$. Then, $\textbf{a}M' = \textbf{0}$.
\end{theorem}
\begin{proof}
    First, define $S = \{T_{ij} | (i, j) \in E\}$. For every equivalence class $K \in S_n/S$, by \Cref{eqparity}, we have $|K \cap S(M_1)| \equiv |K \cap S(M_{k+1})| \pmod 2$ since non-middle gates in the sequence do not affect the set of rivers. 
    
    Let $S' = (S(M_{k+1}) \setminus S(M_1)) \cup (S(M_1) \setminus S(M_{k+1}))$ be obtained by either removing or adding $e$ to $S(M_{k+1})$ depending on whether $e \in S(M_{k+1})$. Let $a_{ij}$ denote the number of rivers $\sigma \in S'$ such that the $j$th term of $\mathcal{L}(\sigma)$ is equal to $i$. Since each individual labeling must appear an even number of times in $S'$, $a_{ij}$ is even for all $i, j$. 
    
    Now, for each permutation $\sigma$, note that the $j$th entry of $\sigma^{-1}$, which we denote $u$, represents the position of $\sigma$ in which $j$ appears. Furthermore, the $j$th entry of $\mathcal{L}(\sigma)$, which we denote $v$, represents the minimum position of $j$ over all permutations in the equivalence class of $\sigma$. From the proof of \Cref{labeling}, $v$ is the minimum vertex in the connected component of $u$ in $G$.

    Thus each $a_{ij} \neq 0$ counts exactly the number of rivers $\sigma \in S'$ such that $\sigma^{-1}(j)$ lies in the connected component of $i$ in $G$. Notably, $a_{ij} = 0$ whenever $i$ is not the smallest vertex in its connected component.

    Let us further define $b_{ij}$ to be the number of rivers $\sigma \in S'$ such that $\sigma^{-1}(j) = i$. Then

    $$a_{ij} = \begin{cases}\sum \limits_{k \in A} b_{kj} & i = \min A, \\ 0 & \textrm{otherwise,}\end{cases}$$

    \noindent where $A$ is the connected component of $i$ in $G$.

    Now, since every $a_{ij}$ is even, for any connected component $A$ of $G$, we must have that $\sum_{k \in A} b_{kj}$ is even. Notice that in the matrix $M'$, the $(k, j)$ entry denotes exactly the parity of the number of rivers $\sigma \in S'$ such that $\sigma^{-1}(j) = k$, so $M'_{k, j} = b_{kj}$. Hence, the sum of the rows of $M'$ whose indices are in $A$ is exactly $\mathbf{0}$. The result follows. 
    
\end{proof}

Although \Cref{abstractsum1} gives a condition for the connected components of $G$ that can be used to lower bound the number of middle CNOT gates, the time needed to compute $M'$ is $O(n \cdot n!)$ as it would be necessary to loop through every river of $M$. We now introduce a shortcut that decreases this time to $O(n^{\omega})$. 

\begin{lemma}
    \label{fastcperfect}
    Let $M$ be an $n \times n$ reversible binary matrix. Then $M \land M^{-\top} + \textbf{I} = \textbf{I} + \sum_{\sigma \in S(M)} P_{\sigma}$. 
\end{lemma}

\begin{proof}

    It is well known that $M^{-1} = \frac{1}{\det M} C^{\top}$ where $C$ is the cofactor matrix of $M$. Thus, the cofactor matrix of $M$ can be computed as $M^{-\top}$ when $M$ is a reversible binary matrix. We say a river $\sigma \in S(M)$ passes through the $(i, j)$-entry of $M$ if $\sigma(i) = j$. Notice that rivers can only pass through nonzero entries, and the parity of the number of rivers passing through any such entry is equal to the $(i, j)$-cofactor. Therefore the parity of the number of rivers in $S(M)$ passing through any given square is given by the matrix $M \land M^{-\top}$. Finally, we add $\mathbf{I}$ to balance the equation.
\end{proof}

Combining \Cref{abstractsum1} and \Cref{fastcperfect} gives us an $O(n^{\omega})$ algorithm shedding light on the connected components formed by the middle CNOT gates. This yields a straightforward lower bound on the number of middle gates.

\begin{theorem}[$c_{\textnormal{perfect}}$ algorithm]
    \label{basiccperfectbound}
    Given $n \times n$ matrix $M$ and a CNOT synthesis $M_1, M_2, \ldots, M_{k+1}$ of $M$, let $M' = M \land M^{-\top} + \mathbf{I}$. Let $G = (V, E)$ a graph such that $V = [n]$ and $(i, j) \in E$ if and only if there is some middle gate involved in the CNOT synthesis adding between rows $i, j$. Then $G$ contains at most $c_{\textnormal{perfect}}(M) := \frac{n + 2 \cdot \Emp M' + \Dup M'}{3}$ connected components.
\end{theorem}

\begin{proof}

    By \Cref{abstractsum1}, the sum of the rows of $M'$ corresponding to each connected component is zero. There are at most $\Emp M'$ rows contained in connected components of size $1$ and at most $\Emp M' + 2\cdot\Dup M'$ rows contained in connected components of size $\leq 2$. Let $x$ be the number of size $1$ connected components and let $y$ be the number of size $2$ connected components. Then there must be at most $x + y + \frac{n - x - 2y}{3} = \frac{2x + y + n}{3}$ connected components in $G$. Since $x + 2y \leq \Emp M' + 2\cdot\Dup M'$ and $x \leq \Emp M'$,

    \begin{align*}\frac{2x + y + n}{3} &= \frac{3x + (x+2y) + 2n}{6} \\ &\leq \frac{3 \cdot \Emp M' + (\Emp M' + 2 \cdot \Dup M') + 2n}{6} \\ &= \frac{n + 2 \cdot \Emp M' + \Dup M'}{3}.\end{align*}
\end{proof}

We remark that we can find $\Emp M'$ and $\Dup M'$ in $O(n^2)$ time by sorting the rows of $M'$.

\begin{corollary}
    \label{cperfectfinal}
    The number of middle gates in any CNOT decomposition of $n \times n$ matrix $M$ is at least $n - c_{\textnormal{perfect}}(M)$. Furthermore, $n - c_{\textnormal{perfect}}(M)$ can be computed in $O(n^{\omega})$ time.
\end{corollary}

\section{The LMC Bound}
\label{mainalgo}

We now combine the lower bounds found in Sections \ref{allgates}, \ref{middlegates} alongside other minor optimizations. We begin by proving that certain matrix operations preserve the value of $s(M)$.

\begin{theorem}
    \label{sizeequivalence}
    For any $n \times n$ reversible binary matrix $M$ and any $n \times n$ permutation matrix $P$, we have $s(M) = s(M^{-1}) = s(M^{\top}) = s(P^{-1}MP)$.
\end{theorem}
\begin{proof}
    Consider a CNOT decomposition $M = R_1R_2\cdots R_{s(M)}$. Taking the inverse of $M$, we obtain a CNOT decomposition $M^{-1} = R_{s(M)}\cdots R_2R_1$. By symmetry, $s(M) = s(M^{-1})$. Taking the transpose of $M$, we obtain a CNOT decomposition $M^{\top} = R_{s(M)}^{\top}\cdots R_2^{\top}R_1^{\top}$. By symmetry, $s(M) = s(M^{\top})$. Finally, conjugating $M$ by $P$, we obtain a CNOT decomposition $P^{-1}MP = R'_1R'_2\cdots R'_{s(M)}$ where $R'_i = P^{-1}R_iP$ is a CNOT gate. By symmetry, $s(M) = s(P^{-1}MP)$, completing the proof.
\end{proof}

Next, we show the number of zeroes on the main diagonal of $M$ gives a second lower bound on the total number of non-link gates in any CNOT decomposition of $M$.

\begin{theorem}
    \label{zeroes}
    Let $M_1, M_2, \ldots, M_{k+1}$ be a CNOT synthesis of $n \times n$ matrix $M$. If $M$ has $\#_0(M)$ zeroes on its main diagonal, then there must be at least $\#_0(M)$ CNOT operations involved in the synthesis which are not link gates.
\end{theorem}

\begin{proof}
    In the identity, there are no zeroes on the main diagonal. Since any CNOT operation changes only one row, we can add at most one zero at a time to the main diagonal. Thus it suffices to show that link gates don't affect the main
    diagonal.

    In the proof of \Cref{linkdec}, we showed that for any link gate adding row $i$ to row $j$, we must have $i, j$ in different connected components of $G_v(M)$. This implies that $M_{i,j} = 0$, so the element $M_{j,j}$ on the main diagonal will not change. Since link gates cannot change the diagonal, we need at least $\#_0(M)$ other CNOT operations in order to synthesize $M$, as desired.
\end{proof}

Finally, we obtain the following lower bound on the size of CNOT circuits. 

\begin{theorem}[LMC Bound]
    \label{finalboss}
    For an $n \times n$ reversible binary matrix $M$, compute $\ell = n - v(M)$, \\ $m = n-\min \{ c_{\textnormal{perfect}}(M), c_{\textnormal{perfect}}(M^{\top}) \}$, and $c = e(M)-v(M)$. Then $$s(M) \geq \ell + \max \left\{m + c, \#_0(M), \#_0(M^{-1}) \right\}.$$
\end{theorem}

\begin{proof}
    By \Cref{lmc}, the sets of link, middle, and cut gates are disjoint. From \Cref{linksandcuts}, the number of link gates is at least $n-v(M)$ and the number of cut gates is at least $e(M) - v(M)$. From \Cref{cperfectfinal}, the number of middle gates is at least $n-c_{\textrm{perfect}}(M)$. By \Cref{zeroes}, at least $\#_0(M)$ non-link gates are required. Combining these results shows that $s(M) \geq \ell + \max \{ n - c_{\textrm{perfect}}(M) + c, \#_0(M)\}$. Since $s(M) = s(M^{-1}) = s(M^{\top}) = s(M^{-\top})$ by \Cref{sizeequivalence}, the result follows from taking the maximum of these lower bounds.
\end{proof}

\subsection{Cycle and Permutation Matrices}

We now apply \Cref{finalboss} to obtain exact values on the size of all $n$-cycle and permutation matrices, as well as the exact depth of $n$-cycle matrices for $n \geq 7$.

\begin{theorem}
    \label{ncycle}
    For any $n$-cycle $\sigma \in S_n$, $s(P_{\sigma}) = 3(n-1)$. Furthermore, let $M_1, M_2, \ldots, M_{3n - 2}$ be a CNOT synthesis of $P_{\sigma}$. Then the sets of link, middle, and cut gates involved in the CNOT synthesis each have size $n - 1$ and form a spanning tree on the set of qubits.
\end{theorem}

\begin{proof}
    Applying \Cref{finalboss}, we note that $v(P_{\sigma}) = 1$ and $e(P_{\sigma}) = n$. Applying \Cref{abstractsum1} and \Cref{fastcperfect}, we have $P'_{\sigma} = P_{\sigma} \land P_{\sigma}^{-\top} + \mathbf{I} = P_{\sigma} + \mathbf{I}$. In particular, the only vector $v$ such that $vP'_{\sigma} = \mathbf{0}$ is $\mathbf{1}$, so the graph $G = (V, E)$, where $V = [n]$ and $(i, j) \in E$ if and only if there is a middle gate adding between rows $i, j$, must be connected. Thus $c_{\textrm{perfect}}(M) = 1$ and $s(P_{\sigma}) \geq 3(n-1)$. \Cref{ncycleexamples} gives various constructions of $s(P_{\sigma}) = 3(n-1)$ alongside the spanning tree structures of the link, middle, and cut gates respectively. We add three original constructions which do not use the swapping method in addition to those given in \cite{bataille2022quantum,liu2024realization,moore2001parallel}.

    \begin{figure*}[!hbtp]
        \centering
        \scalebox{1.25}{
        \begin{tabular}{c|c|c|c|c}
            Pseudocode & L, M, C ordering & \; Link \; & \; Middle \; & \; Cut \; \\ \hline
            \scalebox{0.7}{
            \begin{minipage}{.25\textwidth}
            \begin{algorithm}[H]
            \begin{algorithmic}[0]
            \For{$i = 2$ \textbf{to} $n$}
                \State CNOT($1$, $i$)
                \State CNOT($i$, $1$)
                \State CNOT($1$, $i$)
            \EndFor
            \end{algorithmic}
            \end{algorithm}
            \end{minipage} }
            & $\overbrace{\textrm{LMCLMC}\cdots\textrm{LMC}}^{3(n-1)}$ & Star & Star & Star \\ \hline
            \scalebox{0.7}{
            \begin{minipage}{.25\textwidth}
            \begin{algorithm}[H]
            \begin{algorithmic}[0]
            \For{$i = 1$ \textbf{to} $n-1$}
                \State CNOT($n-i$, $n-i+1$)
                \State CNOT($n-i+1$, $n-i$)
                \State CNOT($n-i$, $n-i+1$)
            \EndFor
            \end{algorithmic}
            \end{algorithm}
            \end{minipage} 
            }
            & $\overbrace{\textrm{LMCLMC}\cdots\textrm{LMC}}^{3(n-1)}$ & Path & Path & Path \\ \hline

            \scalebox{0.7}{
            \begin{minipage}{.25\textwidth}
            \begin{algorithm}[H]
            \begin{algorithmic}[0]
            \For{$i = 1$ \textbf{to} $n-1$}
                \State CNOT($i$, $i+1$)
            \EndFor
            \For{$i = 1$ \textbf{to} $n-1$}
                \State CNOT($i+1$, $i$)
            \EndFor
            \For{$i = 1$ \textbf{to} $n-1$}
                \State CNOT($i$, $n$)
            \EndFor
            \end{algorithmic}
            \end{algorithm}
            \end{minipage} 
            }
            & \; $\overbrace{\textrm{LL}\cdots \textrm{L}}^{n-1}\overbrace{\textrm{MM}\cdots \textrm{M}}^{n-1}\overbrace{\textrm{CC}\cdots \textrm{C}}^{n-1}$ \; & Star & Path & Star \\ \hline

            \scalebox{0.7}{
            \begin{minipage}{.25\textwidth}
            \begin{algorithm}[H]
            \begin{algorithmic}[0]
            \For{$i = 1$ \textbf{to} $n-1$}
                \State CNOT($n-i$, $n$)
                \State CNOT($n$, $n-i$)
            \EndFor
            \State CNOT($1$, $n$)
            \For{$i = 1$ \textbf{to} $n-2$}
                \State CNOT($i+1$, $i$)
            \EndFor
            \end{algorithmic}
            \end{algorithm}
            \end{minipage} 
            }
            & $\overbrace{\textrm{LMLM}\cdots \textrm{LM}}^{2(n-1)}\overbrace{\textrm{CC}\cdots \textrm{C}}^{n-1}$ & Star & Star & Path \\ \hline

            \scalebox{0.7}{
            \begin{minipage}{.25\textwidth}
            \begin{algorithm}[H]
            \begin{algorithmic}[0]
            \For{$i = 2$ \textbf{to} $n$}
                \State CNOT($1$, $i$)
            \EndFor
            \For{$i = 1$ \textbf{to} $n-1$}
                \State CNOT($i+1$, $i$)
                \State CNOT($i$, $i+1$)
            \EndFor
            \end{algorithmic}
            \end{algorithm}
            \end{minipage} 
            }
            & $\overbrace{\textrm{LL}\cdots \textrm{L}}^{n-1}\overbrace{\textrm{MCMC}\cdots \textrm{MC}}^{2(n-1)}$ & Star & Path & Path \\
        \end{tabular}}
        \caption{Table of various constructions to synthesize an $n$-cycle circuit in $3(n-1)$ CNOT operations, alongside their decomposition into link, middle, and cut gates.}
        \label{ncycleexamples}
    \end{figure*}
    
    Therefore $s(P_{\sigma}) = 3(n-1)$. Now, consider a CNOT synthesis $M_1, M_2, \ldots, M_{3n-2}$ of $P_{\sigma}$. There must be exactly $n-1$ link, middle, and cut gates respectively. Because $G$ is connected, the set of middle gates forms a spanning tree on the set of qubits. Now assume for the sake of contradiction that the set of link gates does not form a spanning tree on the set of qubits. Then there is a partition $[n] = A \sqcup B$ such that there is never a CNOT operation between a row in $A$ and a row in $B$, as the first such gate would be a link gate. However, this implies that $P_{\sigma}$ is separable, giving contradiction. 

    Finally, assume for the sake of contradiction that the set of cut gates does not form a spanning tree over the $n$ qubits. Then there exists a partition $[n] = A \sqcup B$ such that there is no cut gate between any $a \in A$ and $b \in B$. There must be a link gate adding a row $a \in A$ to a row $b \in B$, as the link gates form a spanning tree. This means that at some point, $R_a, R_b$ are connected by this link gate. However, if there is no cut gate between $A$ and $B$, then there will always exist a row of $A$ and a row of $B$ that are connected in $G_e$ after this point, giving contradiction, as no two rows are connected in $G_e(P_{\sigma})$.
\end{proof}

Moore and Nilsson \cite{moore2001parallel} showed that the depth of any permutation circuit was at most $6$ by giving a construction using parallel CNOT gates. We demonstrate that this number is optimal for large $n$-cycle circuits.

\begin{corollary}
    Let $n \geq 7$. Then $\max_{\sigma \in S_n} d(P_{\sigma}) = 6$, and $d(P_{\sigma}) = 6$ when $\sigma$ is an $n$-cycle.
\end{corollary}

\begin{proof}
Note that any parallel CNOT operation can be decomposed into at most \(\left\lfloor \frac{n}{2} \right\rfloor\) CNOT gates. Thus, we have 
\[
d(M) \geq \left\lceil \frac{s(M)}{\left\lfloor \frac{n}{2} \right\rfloor } \right\rceil
\]
for any \(n \times n\) reversible binary matrix \(M\). Hence, for \(n \geq 7\), any \(n\)-cycle \(\sigma\) satisfies \(d(P_{\sigma}) \geq 6\). Finally, the construction given in \cite{moore2001parallel} shows that \(d(M) \leq 6\) for any choice of \(P_{\sigma}\).
\end{proof}

\begin{theorem}
    For any $\sigma \in S_n$, we have $s(P_{\sigma}) = 3(n-k)$, where $k$ is the number of disjoint cycles in $\sigma$.
\end{theorem}

\begin{proof}
    Using the construction from \Cref{ncycle} on each disjoint cycle of $\sigma$, we achieve a synthesis of $P_{\sigma}$ in $3(n-k)$ CNOT gates, showing that $s(P_{\sigma}) \leq 3(n-k)$.

    Now for the sake of contradiction assume that there exists some $\sigma \in S_n$ with minimal $k$ that satisfies $s(P_{\sigma})<3(n-k)$. \Cref{ncycle} implies that $k \geq 2$. Now, choose qubits $i, j$ in disjoint cycles of $\sigma$. Appending $\textrm{CNOT}(i, j), \textrm{CNOT}(j, i), \textrm{CNOT}(i, j)$ in that order to the end of the circuit synthesizing $P_{\sigma}$ has the effect of swapping rows $i, j$. However, as illustrated in \Cref{cyclejoin}, this synthesizes another permutation matrix $P_{\sigma'}$ such that $\sigma'$ has one fewer disjoint cycle than $\sigma$ and additionally $s(P_{\sigma'}) < 3(n-k+1)$. This either contradicts the minimality of $k$ or \Cref{ncycle}, finishing. 
\end{proof}

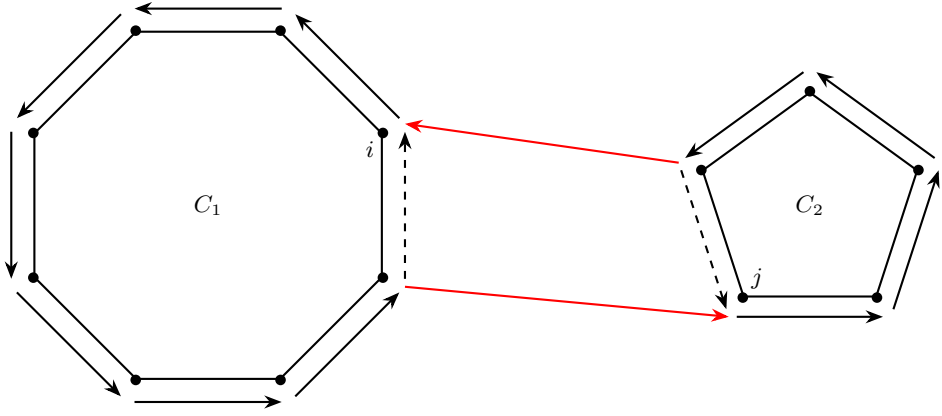
\begin{figure*}[!hbtp]
    \begin{tikzpicture}
    \usetikzlibrary{shapes.geometric}
    \usetikzlibrary{arrows.meta}
    \tikzset{myarrow/.style={-Stealth}}
    
    \node[regular polygon,draw, thick,regular polygon sides = 8, minimum size = 5 cm] (a) at (0,0) {};

    \node[regular polygon,draw=none,regular polygon sides = 8, minimum size = 5.65 cm] (c) at (0,0) {};

    \draw[myarrow,thick, shorten >=3pt, shorten <=3pt] (c.corner 1) -- (c.corner 2);
    \draw[myarrow,thick, shorten >=3pt, shorten <=3pt] (c.corner 2) -- (c.corner 3);
    \draw[myarrow,thick, shorten >=3pt, shorten <=3pt] (c.corner 3) -- (c.corner 4);
    \draw[myarrow,thick, shorten >=3pt, shorten <=3pt] (c.corner 4) -- (c.corner 5);
    \draw[myarrow,thick, shorten >=3pt, shorten <=3pt] (c.corner 5) -- (c.corner 6);
    \draw[myarrow,thick, shorten >=3pt, shorten <=3pt] (c.corner 6) -- (c.corner 7);
    \draw[myarrow,dashed,thick, shorten >=3pt, shorten <=3pt] (c.corner 7) -- (c.corner 8);
    \draw[myarrow,thick, shorten >=3pt, shorten <=3pt] (c.corner 8) -- (c.corner 1);

    \node[regular polygon,draw, thick, minimum size = 3 cm] (b) at (8,0) {};

    \node[regular polygon,draw=none, minimum size = 3.65 cm] (d) at (8,0) {};

    \foreach \x in {1,2,...,8}
    \fill (a.corner \x) circle[radius=2pt];

    \foreach \x in {1,2,...,5}
    \fill (b.corner \x) circle[radius=2pt];

    \draw[myarrow,thick, shorten >=3pt, shorten <=3pt] (d.corner 1) -- (d.corner 2);
    \draw[myarrow,dashed,thick, shorten >=3pt, shorten <=3pt] (d.corner 2) -- (d.corner 3);
    \draw[myarrow,thick, shorten >=3pt, shorten <=3pt] (d.corner 3) -- (d.corner 4);
    \draw[myarrow,thick, shorten >=3pt, shorten <=3pt] (d.corner 4) -- (d.corner 5);
    \draw[myarrow,thick, shorten >=3pt, shorten <=3pt] (d.corner 5) -- (d.corner 1);

    \draw[myarrow,red,thick] (c.corner 7) -- (d.corner 3);
    \draw[myarrow,red,thick] (d.corner 2) -- (c.corner 8);

    \node at (0,0) {$C_1$};
    \node at (8,0) {$C_2$};
    \node[anchor = north east] at (a.corner 8) {$i$};
    \node[anchor = south west] at (b.corner 3) {$j$};
    \end{tikzpicture}
    \centering
    \caption{The effect of adding a swap operation (3 CNOT gates) between two disjoint cycles.}
    \label{cyclejoin}

\end{figure*}

\section{Linkability}
\label{linkability}

In this section, we introduce the linkability problem, a subclass of the global minimization problem proposed in \cite{jiang2020optimal}. Specifically, we determine whether it is possible to synthesize an arbitrary circuit $M$ with $v(M) = 1$ in at most $n - 1$ gates in $O(n^{\omega})$ time. We begin by introducing a few notions used in the algorithm.

\begin{lemma}
    Consider a CNOT decomposition of $n \times n$ matrix $M$. If every CNOT gate is a link gate, then the influence graph $\mathcal{I}(M)$ forms a directed acyclic graph. 
\end{lemma}

\begin{proof}
    Assume that there is a directed cycle $x_1 \to x_2 \to \cdots \to x_m \to x_1$ in $\mathcal{I}(M)$. Let indices be taken modulo $m$. The permutation 
    $$\sigma(x) = \begin{cases}
    x_{i+1} & x = x_i, i \in [m], \\
    x & \textrm{otherwise,}
    \end{cases}$$
     satisfies $\sigma \in S(M)$. However, this means $S(M) \neq S(\textbf{I})$ while link gates do not affect the set of rivers by \Cref{lmc}, giving contradiction.
\end{proof}

We recall the following result in graph theory, which gives us the time complexity for computing the transitive reduction of an arbitrary graph.

\begin{theorem}[Aho et al. \cite{aho1972transitive} and Furman \cite{furman1970application}]
    The transitive reduction of a directed graph $G$ with $n$ vertices can be computed in $O(n^{\omega})$ time.
\end{theorem}

Now, we present the algorithm for the linkability problem in its entirety.

\begin{theorem}
\label{reversibility}
    Given an $n \times n$ reversible binary matrix $M$ such that $v(M) = 1$, it is possible to determine whether $s(M) \leq n - 1$ in $O(n^\omega)$ time.
\end{theorem}

\begin{proof} \;

\begin{steps}
\item  We begin by assuming that $M$ satisfies $s(M) \leq n - 1$. Then by \Cref{finalboss}, $s(M) = n-1$. By assumption, there exists a CNOT synthesis $M_1, M_2, \ldots, M_{n}$ of $M$. It is clear that all the CNOT gates must be link gates since $v(M) = 1$. We compute the influence graph $\mathcal{I}(M) = (V, A)$, which takes $O(n^2)$ time. 

\item  Let $\mathcal{I'} = (V, A')$ be the transitive reduction of $\mathcal{I}$. We claim that for each $(i, j) \in A'$, one of the $n-1$ CNOT operations used in the synthesis of $M$ must add row $i$ to row $j$. We prove this as follows. If $(i, j) \in A'$, there must not have been a directed path from $i$ to $j$ of length $\geq 2$ in $\mathcal{I}$. Assume that no CNOT operation adds row $i$ to row $j$. Then, there exists a sequence $k_1, k_2, \ldots, k_s$ with $k_1 = i$ and $k_s = j$ such that there is a CNOT gate adding row $k_i$ to $k_{i+1}$ for all $i \in [s-1]$. However, this implies each $(k_i, k_{i+1}) \in A$ as link gates cannot remove nonzero entries. This is a directed path of length $\geq 2$ from $i$ to $j$, giving contradiction and completing the proof.

Hence $\mathcal{I'}$ both maintains connectivity and has at most $n-1$ edges, so it must be a spanning tree. The exact set of CNOT gates can thus be read off of $A'$.

 \item Now, we determine the order in which these CNOT operations took place. For any 3 vertices $(u, v, w)$ in $\mathcal{I'}$ where $(u, v), (v, w) \in A'$, it is easy to see that $M_{w, u} = 1$ if and only if $\textrm{CNOT}(u, v)$ happens before $\textrm{CNOT}(v, w)$, as $(u, w) \notin A'$ and no other directed path from $u$ to $w$ exists in $\mathcal{I'}$. 

\item From the previous step, we know the relative ordering of all adjacent CNOT gates. Any ordering of the $n-1$ CNOT gates respecting these relative orders synthesizes the same matrix. Thus, we can pick any such full ordering and perform the CNOT operations step by step on $\textbf{I}$ to yield a matrix $N$. This entire process requires $O(n^2)$ time. If $M = N$, we have just found a construction to obtain $M$ in exactly $n-1$ CNOT gates. Otherwise, if $M \neq N$, this is a contradiction on the original assumption that $s(M) \leq n - 1$, finishing.
\end{steps}

\vspace{0.5cm}
The overall time complexity of this algorithm is $O(n^2) + O(n^\omega) + O(n^2) + O(n^2) = O(n^\omega)$, as desired.
\end{proof}

\section{Results}
\label{results}
For $n \leq 5$, the size of all linear reversible CNOT circuits can be easily found using a BFS algorithm over the set of all $n$-qubit circuits. In \Cref{accuracy}, we provide various metrics on how well the lower bound on size from \Cref{finalboss} matches the exact size for $n$-qubit circuits. In \Cref{heatmaps}, we provide heatmaps for the confusion matrix between the lower bound and the exact size. The raw data used can be found in Appendix \ref{n5}.

\begin{table}[!hbtp]
\centering
{\renewcommand{\arraystretch}{1.2}
\scalebox{1.1}{

\begin{tabular}{|c|c|c|c|}
\hline 
\;Metric\; & \;$n \leq 3$\; & \;$n = 4$\; & \;$n = 5$\; \\ 
\hline
$\Delta = 0$ & $100 \%$ & $67.7 \%$ & $23.1 \%$ \\
$\Delta \leq 1$ &  $100 \% $ & $99.5 \%$ & $83.0 \%$ \\
$\Delta \leq 2$  & $100 \%$ & $100 \%$ & $99.7 \%$ \\
$\sigma$ &  0 & 0.581 & 1.137\\
$MAD$ &  0 & 0.328 & 0.941\\
$R^2$  & 1 & 0.811 & 0.649 \\
$PCC$  &  1 & 0.901 & 0.806 \\
\hline
\end{tabular}
}
}
\caption{Various metrics on \Cref{finalboss} for $n$-qubit circuits (here, $\Delta$ denotes the deviation between lower bound and actual size). }
\label{accuracy}
\end{table}

\begin{figure*}[!hbtp]
\centering
\begin{subfigure}[b]{0.25\textwidth}
\centering
    \includegraphics[width=\textwidth]{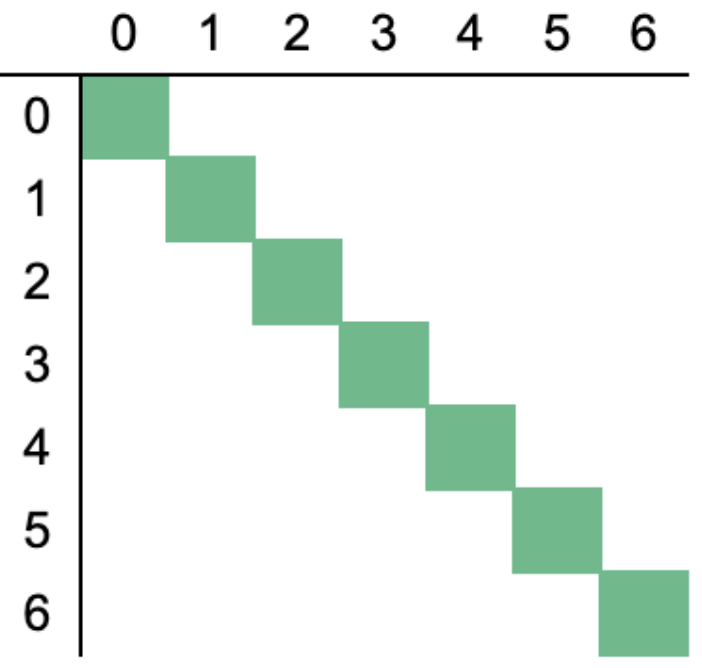}
\end{subfigure}
\hfill
\begin{subfigure}[b]{0.25\textwidth}
\centering
    \includegraphics[width=\textwidth]{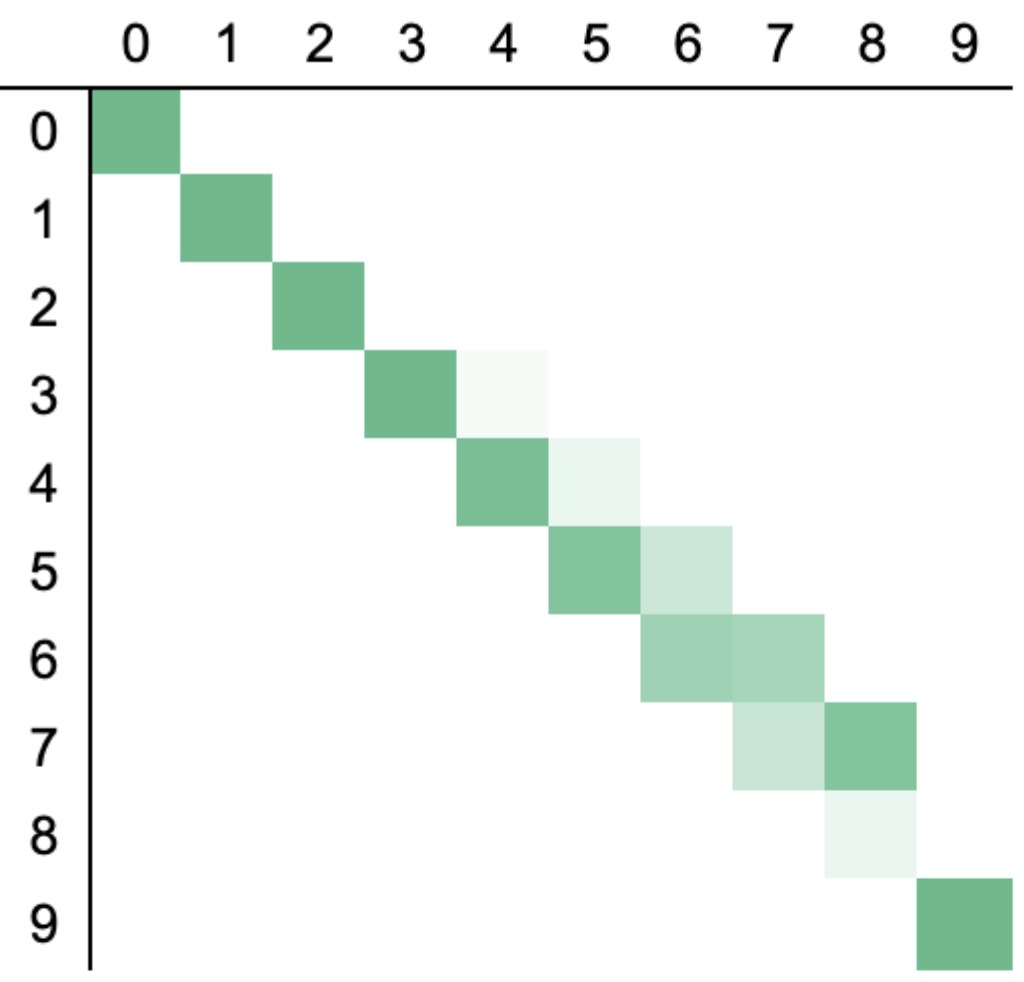}
\end{subfigure}
\hfill
\begin{subfigure}[b]{0.25\textwidth}
\centering
    \includegraphics[width=\textwidth]{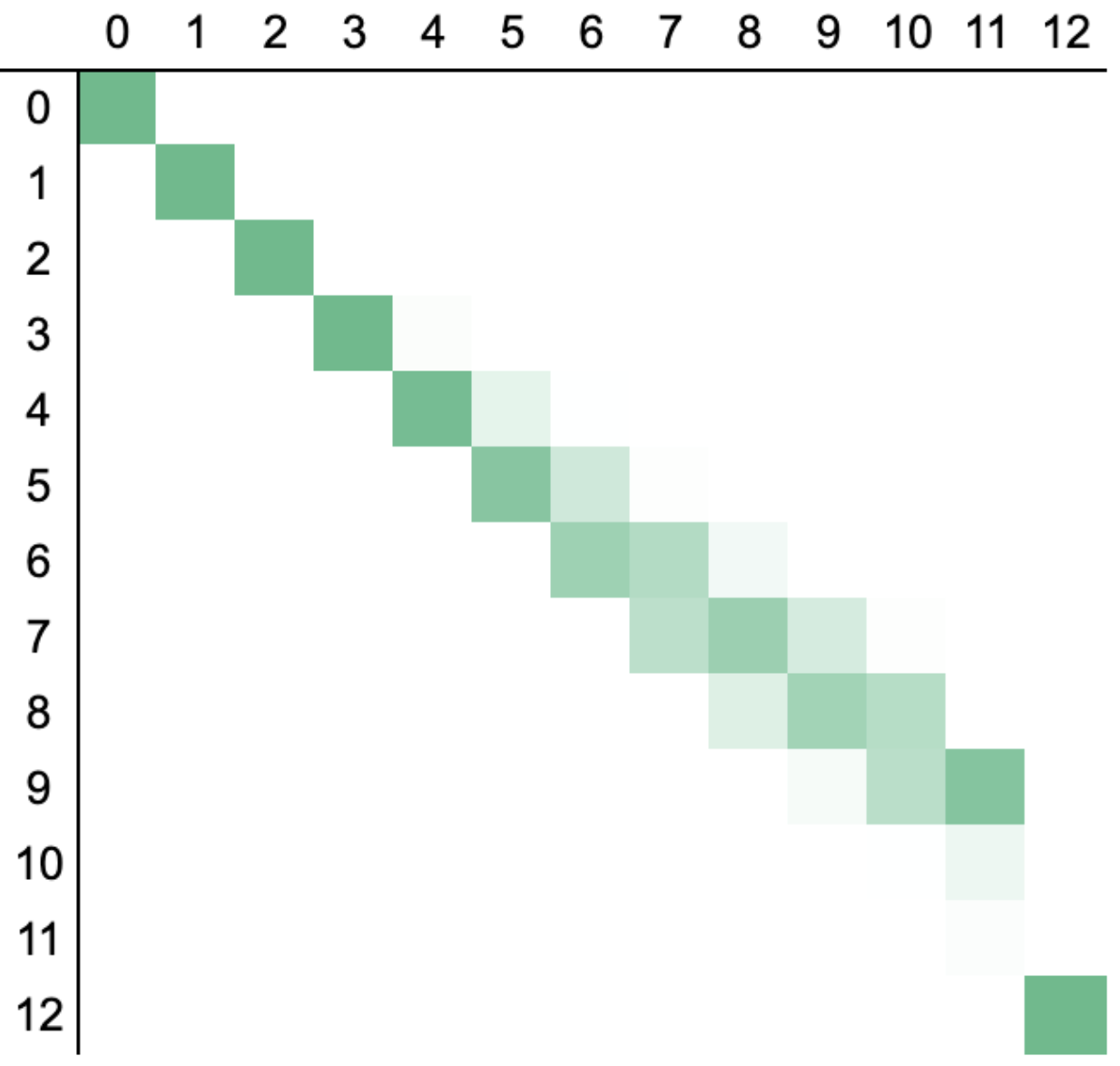}
\end{subfigure}
\caption{Column-by-column heatmaps for $n = 3, 4, 5$ between the lower bound in \Cref{finalboss} (on the vertical) and the size of the circuit (on the horizontal).}
\label{heatmaps}
\end{figure*}

Although our lower bound performs well at various metrics for $n \leq 5$, Jiang et al. show in \cite{jiang2020optimal} that as $n$ grows larger, the average size of an arbitrary $n \times n$ binary reversible circuit grows to $O(\frac{n^2}{\log n})$ while our lower bound can be at most $3(n-1)$. This means that although our lower bound still holds for large $n$, it cannot be used as a size-approximation algorithm for large $n$ without further improvements.

\section{Conclusion and Future Directions}
\label{conclusion}

In this paper, we introduce a categorization of the gates of CNOT decompositions into link, middle, and cut gates. Under this framework, we derive lower bounds on the number of CNOT gates in each category and combine them to obtain a lower bound on the total number of CNOT gates needed to synthesize a given circuit that avoids brute-force search and can be obtained in $O(n^{\omega})$ time. We apply our framework to the problem of synthesizing the $n$-cycle matrix to show that known constructions with $3(n-1)$ CNOT gates are optimal, and further show the size of any permutation matrix is $3(n-k)$, where $k$ is the number of disjoint cycles in the permutation. 

 We introduce an $O(n^\omega)$ time algorithm determining whether a given circuit with $v(M) = 1$ can be synthesized in fewer than $n$ CNOT operations, resolving a subclass of the minimum circuit size problem. We hope the methods and ideas developed in this paper will aid the search for further polynomial-time lower bounds on the size of CNOT circuits, as well as demonstrate the optimality or near-optimality of known CNOT decompositions.

To end, we present several interesting future directions on both the minimum circuit size problem as well as for improving our lower bound in \Cref{finalboss}.

\begin{problem}
    In \cite{patel2003efficient}, an information theoretic bound is used to show that the maximum size of a circuit on $n$ qubits is at least $\frac{n^2}{2\log_2 n}$. However, in the cases $n \leq 5$, the maximum size of an $n$-qubit circuit is always $3(n-1)$, attained by the $n$-cycle circuit. This is plausible for all $n \leq 27$, as we have $\frac{n^2}{2\log_2 n} \leq 3(n-1)$. At what point does the asymptotic bound overtake the $3(n-1)$ pattern? 
\end{problem}
\begin{problem}
 In the $n=5$ case, a single circuit, alongside those equivalent to it through \Cref{sizeequivalence}, lies far off the main diagonal in the table in Appendix \ref{n5}. It turns out that $M \land M^{-\top} + \mathbf{I} = \mathbf{0}$ (see \Cref{fig:mainfig}), rendering the $c_{\textrm{perfect}}$ bound ineffective. Can we improve the $c_{\textrm{perfect}}$ algorithm to account for this ``$c_{\textrm{perfect}}$ glitch" and obtain a more powerful lower bound on the number of middle CNOT gates?

 \begin{figure*}[!hbtp]
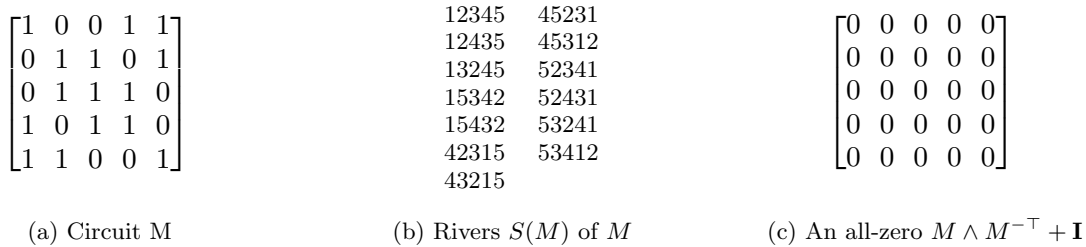

    \centering
    \begin{subfigure}[b]{0.3\textwidth}
        \centering
        \scalebox{1.2}{
        $\begin{bmatrix}
            1 & 0 & 0 & 1 & 1 \\
            0 & 1 & 1 & 0 & 1 \\
            0 & 1 & 1 & 1 & 0 \\
            1 & 0 & 1 & 1 & 0 \\
            1 & 1 & 0 & 0 & 1
        \end{bmatrix}$
        }
        \vspace{0.35cm}
        \caption{Circuit M}
    \end{subfigure}
     % Ensure that subfigures are spaced out equally
    \begin{subfigure}[b]{0.3\textwidth}
        \begin{itemize}[itemsep=0pt, parsep=0pt, label={}]
            \item \;\, \qquad 12345 \quad 45231    
            \item  \;\, \qquad 12435 \quad 45312
            \item \;\, \qquad 13245 \quad 52341
            \item \;\, \qquad 15342 \quad 52431
            \item \;\, \qquad 15432 \quad 53241
            \item \;\, \qquad 42315 \quad 53412
            \item \;\, \qquad 43215 \quad 
        \end{itemize}
        \caption{Rivers $S(M)$ of $M$}
    \end{subfigure}
     % This will add horizontal space between subfigures
    \begin{subfigure}[b]{0.3\textwidth}
        \centering
        \scalebox{1.2} {
        $\begin{bmatrix}
            0 & 0 & 0 & 0 & 0 \\ 
            0 & 0 & 0 & 0 & 0 \\ 
            0 & 0 & 0 & 0 & 0 \\ 
            0 & 0 & 0 & 0 & 0 \\ 
            0 & 0 & 0 & 0 & 0 
        \end{bmatrix}$
        }
        \vspace{0.35cm}
        \caption{An all-zero $M \land M^{-\top} + \mathbf{I}$}
    \end{subfigure}
    \caption{The $c_{\textrm{perfect}}$ glitch}
    \label{fig:mainfig}
\end{figure*}
\end{problem}

\begin{problem}
In addition to permutations, are there other classes of circuits for which the bound in \Cref{finalboss} is optimal?
\end{problem}
\begin{problem}
Is it possible to extend our methods to circuits with ancillae, topological constraints, or one-qubit gates such as $T$ gates and Hadamard gates?
\end{problem}

\section{Acknowledgements}
We thank Yunseo Choi and Kevin Cong for their valuable feedback. 

\bibliographystyle{ieeetr} %alpha
\bibliography{apssamp}

\appendix
\section{Performance of \Cref{finalboss} for Circuits with $4$ or $5$ Qubits}
\label{n5}

In this appendix, the lower bound given in \Cref{finalboss} is plotted against the exact size of the circuit in the $4$ and $5$-qubit cases, with the entry in row $i$ and column $j$ denoting the number of linear reversible circuits with lower bound $i$ and size $j$. For $n \leq 3$, equality holds.

\begin{figure*}
{\renewcommand{\arraystretch}{1.2}
\resizebox{10cm}{!}{
\begin{tabular}{c|cccccccccc}
 \qquad &\;\;0\;\;&\;\;1\;\;&\;\;2\;\;&\;\;3\;\;&\;\;4\;\;&\;\;5\;\;&\;\;6\;\;&\;\;7\;\;&\,\;\;8\;\;&\,\,\;\;9\;\;\\
\hline
\; 0 \; & 1 &  &  &  &  &  &  &  &  &  \\
1 &  & 12 &  &  & & & & & &  \\
2 &  &  & 96 & & & & & & &  \\
3 &  &  & & 542 & 138 & & & & & \\
4 &  &  & & & 1920 & 756 & 12 & & & \\
5 &  &  &  &  &  & 4560 & 2589 & 84 & & \\
6 &  &  &  &  &  &  & 4929 & 2464 & & \\
7 &  &  &  &  &  & & & 1510 & 469 & \\
8 &  &  &  &  &  &  &  & & 72& \\
9 &  &  &  &  &  &  &  & & & 6 \\

\end{tabular}
}}

\caption{The distribution of $4$-qubit circuits, with the lower bound in \Cref{finalboss} plotted on the vertical and the size of the circuit plotted on the horizontal.}
\end{figure*}

\begin{figure*}[!hbtp]
\centering
\begin{center}
{\renewcommand{\arraystretch}{0.8}
\resizebox{17cm}{!}{%
\begin{tabularx}{1.0\textwidth} { 
   >{\centering\arraybackslash}X |
   >{\centering\arraybackslash}X 
   >{\centering\arraybackslash}X >{\centering\arraybackslash}X 
   >{\centering\arraybackslash}X >{\centering\arraybackslash}X 
   >{\centering\arraybackslash}X >{\centering\arraybackslash}X 
   >{\centering\arraybackslash}X >{\centering\arraybackslash}X 
   >{\centering\arraybackslash}X >{\centering\arraybackslash}X >{\centering\arraybackslash}X 
   >{\centering\arraybackslash}X }
& \multicolumn{1}{c}{0} & \multicolumn{1}{c}{1} & \multicolumn{1}{c}{2} & \multicolumn{1}{c}{3} & \multicolumn{1}{c}{4} & \multicolumn{1}{c}{5} & \multicolumn{1}{c}{6} & \multicolumn{1}{c}{7} & \multicolumn{1}{c}{8} & \multicolumn{1}{c}{9} & \multicolumn{1}{c}{10} & \multicolumn{1}{c}{11} & \multicolumn{1}{c}{12} \\ \hline
\\[-0.2cm]
0  & 1                     &                       &                       &                       &                       &                       &                       &                       &                       &                       &                        &                        &                        \\ 
1  &                       & 20                    &                       &                       &                       &                       &                       &                       &                       &                       &                        &                        &                        \\ 
2  &                       &                       & 260                   &                       &                       &                       &                       &                       &                       &                       &                        &                        &                        \\ 
3  &                       &                       &                       & 2570                  & 690                   &                       &                       &                       &                       &                       &                        &                        &                        \\ 
4  &                       &                       &                       &                       & 18990                 & 20540                 & 3640                  &                       &                       & 12                    &                        &                        &                        \\ 
5  &                       &                       &                       &                       &                       & 97320                 & 176505                & 37880                 & 1320                  &                       &                        &                        &                        \\ 
6  &                       &                       &                       &                       &                       &                       & 360325                & 917960                & 322750                & 11580                 &                        &                        &                        \\ 
7  &                       &                       &                       &                       &                       &                       &                       & 813870                & 2448985               & 925340                & 15840                  &                        &                        \\ 
8  &                       &                       &                       &                       &                       &                       &                       &                       & 798120                & 2072000               & 365540                 &                        &                        \\ 
9  &                       &                       &                       &                       &                       &                       &                       &                       &                       & 216378                & 350120                 & 13340                  &                        \\ 
10 &                       &                       &                       &                       &                       &                       &                       &                       &                       &                       & 5040                   & 1920                   &                        \\ 
11 &                       &                       &                       &                       &                       &                       &                       &                       &                       &                       &                        & 480                    &                        \\ 
12 &                       &                       &                       &                       &                       &                       &                       &                       &                       &                       &                        &                        & 24 
\end{tabularx}
}
}
\caption{The distribution of $5$-qubit circuits, with the lower bound in \Cref{finalboss} plotted on the vertical and the size of the circuit plotted on the horizontal.}
\end{center}
\end{figure*}

\end{document}